\documentclass[letterpaper, 10 pt, conference]{ieeeconf}  

\IEEEoverridecommandlockouts                              

\usepackage{hyperref}
\usepackage{url}
\usepackage{graphicx}
\usepackage{hyperref}
\usepackage{comment}
  \usepackage{amsthm}
\usepackage{amsmath}
\usepackage{amssymb}
\usepackage{amsfonts}
\usepackage{amsopn}
\usepackage{diagbox}
\usepackage{makecell}
\usepackage{cite}
\usepackage{hyperref}
\hypersetup{
    colorlinks=true,
    linkcolor=blue,
    filecolor=blue,      
    urlcolor=blue,
    citecolor=cyan,
}
\usepackage[dvipsnames]{xcolor}
 
\usepackage{pifont}
\usepackage{booktabs} 
\usepackage{multirow}

\usepackage{amsmath}
\usepackage{bbding}

\usepackage{balance}

\usepackage{mathrsfs}

\usepackage{subcaption}

\usepackage{algorithm}
\usepackage{algpseudocode}
\usepackage{caption}

\newtheorem{theorem}{Theorem}[]
\newtheorem{proposition}{Proposition}[]
\newtheorem{lemma}{Lemma}[]
\newtheorem{remark}{Remark}

\newtheorem{definition}{Definition}[]

\newcommand{\K}{\mathbf{K}}
\newcommand{\G}{\mathbf{G}}

\newcommand{\z}{\mathbf{z}}

\newcommand{\Q}{\mathbf{Q}}

\newcommand{\Rb}{\mathbb{R}}
\newcommand{\bu}{\mathbf{u}}

\newcommand{\x}{\mathbf{x}}

\newcommand{\A}{\mathbf{A}}

\newcommand{\B}{\mathbf{B}}

\newcommand{\D}{\mathbf{D}}

\newcommand{\Eb}{\mathbb{E}}

\newcommand{\F}{\mathbf{F}}

\newcommand{\Hb}{\mathbf{H}}

\newcommand{\I}{\mathbf{I}}

\usepackage{todonotes}
\title{\LARGE \bf
Benefits  of Linear Dynamic State Feedback in Co-stabilization
}

\author{Xiong Zeng$^*$  
  \quad\;\; Necmiye Ozay$^*$ \quad\;\; Mario Sznaier$^+$
	 \thanks{$^*$Department of Electrical Engineering and Computer Science, University
of Michigan, Ann Arbor, MI 48109. Emails: \{zengxion, necmiye\}@umich.edu.}
     \thanks{$^+$Department of Electrical and Computer Engineering,  Northeastern University, Boston, MA 02115. Email: m.sznaier@northeastern.edu.}
		}

\begin{document}

\maketitle
\thispagestyle{empty}
\pagestyle{empty}

\begin{abstract}
Co-stabilization, i.e., designing a single controller that stabilizes multiple systems, is a fundamental problem in robust and data-driven control. While any stabilizable linear system admits a stabilizing linear static state feedback controller, this equivalence does not extend to co-stabilization. In particular, there exist system collections that cannot be co-stabilized by linear static state feedback but can be co-stabilized using linear dynamic state feedback.
In this paper, we study the role of controller memory in co-stabilization. We show that linear dynamic state feedback strictly enlarges the set of co-stabilizable systems compared to static feedback, for both scalar systems and high-dimensional examples. At the same time, we identify structural limitations that cannot be overcome even with dynamic controllers.
We also develop a path-integral-based algorithm for computing co-stabilizing controllers for a finite set of systems. Numerical results demonstrate that increasing controller memory enlarges the feasible co-stabilization region. These results highlight controller architecture as a key structural factor in co-stabilization, with potential implications for reducing the sample complexity of learning-based control.

\end{abstract}

 \section{Introduction}

Learning a stabilizing controller is a key step in many learning-based control tasks. Recent works have studied fundamental limits of learning to stabilize from a statistical perspective \cite{tsiamis2022learning,zeng2023hardness,zhangstabilizing,toso2025learning}. In these works, the hardness of the problem is characterized by sample complexity, i.e., the number of samples required to achieve stabilization with high probability. Of particular relevance is \cite{zeng2023hardness}, which shows that the hardness of learning to stabilize is governed by two key factors: distinguishability and co-stabilizability of the systems involved. While these principles appear broadly applicable, most existing sample complexity results focus on linear systems under full state observation with linear static state feedback controllers.\footnote{While dynamic controllers are used in some learning-based control approaches \cite{dean2020sample,hazan2020nonstochastic}, the specific benefits of controller memory for stabilization have not been systematically investigated to the best of our knowledge.}

Motivated by these observations, we focus on the \emph{co-stabilization} problem, i.e., designing a single controller that stabilizes multiple systems. Co-stabilization is a fundamental structural property underlying robustness and learning-based control: the ability to stabilize a larger set of systems directly translates into increased tolerance to model uncertainty.

A key fact underlying our work is the following. While any stabilizable fully observed linear system admits a stabilizing linear static state feedback controller, this equivalence does not extend to co-stabilization. In particular, there exist pairs of systems that cannot be co-stabilized by any linear static state feedback controller but can be co-stabilized by a linear dynamic state feedback controller. This observation suggests that controller memory can fundamentally enlarge the class of systems that can be co-stabilized by a single policy.
In this paper, we formalize this intuition and develop an algorithm to empirically demonstrate the benefits of dynamic state feedback for co-stabilization.
Specifically, our contributions are summarized as follows:
\begin{itemize}
    \item First, we prove that the linear dynamic state feedback is more expressive on the co-stabilization problem than the linear static state feedback for some typical examples from the literature. This is achieved by using results from strong stabilization and simultaneous stabilization in robust control \cite{doyle2013feedback,sanchez1998robust}. 
    \item Inspired by our theoretical results,  we design a co-stabilization path integral algorithm, which iteratively updates linear dynamic state feedback policies, and empirically verify the advantage of linear dynamic state feedback on co-stabilization.
\end{itemize}

The rest of the paper is organized as follows. 
In Section \ref{sec_problem}, we introduce the problem setup. 
In Section \ref{sec_sta_vs_dyna}, we present our main theoretical results. 
In Section \ref{sec_algorithm_path_integral}, we develop a path-integral-based algorithm for computing co-stabilizing controllers. 
Section \ref{sec_exp} provides numerical experiments. 
Finally, Section \ref{sec_con} concludes the paper and discusses future directions. 
All proofs are deferred to the appendix.

\section{Problem Setup}
\label{sec_problem}
We consider the following fully-observed discrete-time linear time-invariant (LTI) system:
\begin{equation}
\label{LTI}
S_s: \mathbf{x}_{t+1} =\mathbf{A} \mathbf{x}_t + \mathbf{B} \mathbf{u}_t,
\end{equation}
where $\mathbf{x}_t \in \mathbb{R}^{n}$, $\mathbf{u}_t \in \mathbb{R}^{m}$ are the state and input at time $t$.  For simplicity, we assume $ \Eb[\mathbf{x}_{0}\mathbf{x}_{0}^\top] = \I_n$.   
In the remainder of the paper, we denote a system in the form \eqref{LTI} by the pair $(\mathbf{A},\B)$.

 We consider linear dynamic state feedback controllers, represented with a linear time-invariant system of the form:  
\begin{equation}
\label{p_dynamic_state_feedback}
   S_c: \begin{cases}
        \z_{t+1} &= \F \z_t + \G \x_t\\
        \bu_t & = \Hb \z_t + \K  \x_t
    \end{cases},
\end{equation}
where the parameter matrices $\K \in \Rb^{m\times n}$, $\Hb \in \Rb^{m\times p}$, $\G \in \Rb^{p \times n}$, $\F \in \Rb^{p\times p}$, and the state of the controller $\z_t\in\Rb^p$, where $p$ is the memory of the controller.  Then, the following augmented system represents the closed-loop:
\begin{equation}
\label{aug_LTI_dy}
 \begin{bmatrix}
    \x_{t+1}\\
    \z_{t+1}
\end{bmatrix} = \underbrace{\begin{bmatrix}
\A + \B \K & \B \Hb \\
\G & \F
\end{bmatrix}}_{=: \D_{\A,\B}} \begin{bmatrix}
    \x_t\\
    \z_t
\end{bmatrix}.
 \end{equation}

\begin{remark}
\label{rem_relationship}
    When $\Hb=\mathbf{0}_{m\times p},\G=\mathbf{0}_{p \times n},$ and $\F=\mathbf{0}_{p \times p}$, the linear dynamic state feedback controller is reduced to the linear static state feedback controller $\bu_{t} = \K \x_t$. When  $\G=[\mathbf{I}_n \;\mathbf{0}_{ n \times (h-1)n}]^\top$, $\F=\begin{bmatrix}
\mathbf{0}_{n\times(h-1)n } & \mathbf{0}_{n\times n } \\
\mathbf{I}_{(h-1)n} & \mathbf{0}_{(h-1)n\times n}
\end{bmatrix}$, and $ \Hb = [\K_1 \dots   \K_{h-1}]$ with $\K_i \in \Rb^{m\times n}$ for $i\in[h-1]$, the linear dynamic state feedback controller is reduced to the linear {state-history} feedback controller $\bu_t = \K\x_t + \sum_{i=1}^{h-1} \K_i \x_{t-i}$.
\end{remark}

Since a linear static state feedback controller is a special case of a linear dynamic state feedback controller, any task that can be achieved by the former can also be achieved by the latter. On the other hand, when $(\A,\B)$ is given, it is well known that there exists a controller that stabilizes $(\A,\B)$ if and only if there exists a linear static state feedback controller that does so (see, e.g., {\cite[Theorem 14.5]{hespanha2018linear}}). Next, we look at the problem of co-stabilization (also known
as simultaneous stabilization {\cite{blondel1994simultaneous}}).


\begin{definition}[Co-stabilizability by a linear controller]
\label{def_co_stabilization_static}
A pair of systems $S_1 =(\A_1,\B_1)$ and $S_2=(\A_2, \B_2)$ is co-stabilizable by a linear controller $S_c$ if both the feedback interconnection of $S_1$ and $S_c$ and that of $S_2$ and $S_c$ are asymptotically stable. 
\end{definition}

Note that in the case of linear static state feedback controllers, this definition simplifies to existence of $\mathbf{K}$ such that  $\max_{i\in\{1,2\}} \{ \rho(\A_i + \B_i \mathbf{K})\} <1$, where $\rho(\cdot)$ denotes the spectral radius of a square matrix. Similarly, the pair is said to be co-stabilizable by a linear dynamic state feedback if there exist a dimension $p\in \mathbb{Z}_{\geq0}$ and feedback gains $(\K,\Hb,\G,\F)$ such that $\max_{i\in\{1,2\}} \{ \rho(\D_{\A_i,\B_i})\}<1$.






\section{Linear Static State Feedback vs. Linear Dynamic State Feedback}
\label{sec_sta_vs_dyna}

\begin{table}[t]
\centering
\caption{Co-Stabilizability Summary ($|a_1|,|a_2| > 1, b > 0$)}
\label{tab:Co-Stabilizability-Summary}
\renewcommand{\arraystretch}{1.1}
\setlength{\tabcolsep}{2pt}
\begin{tabular}{|c|c|c|c|}
\hline
\textbf{Cases}   & \textbf{Static } & \textbf{Dynamic} \\
\hline
Pair \eqref{eq_two_scalar_sys_gaussian} with $|a_1-a_2|>2,\, b_1=b_2=b$ 
 & \ding{55} (Thm.~\ref{thm_three_static_vs_dynamic}) & \checkmark (Thm.~\ref{thm_three_static_vs_dynamic}) \\
\hline
Pair \eqref{eq_two_scalar_sys_gaussian} with $a_1=a_2=a, \tfrac{b_1}{b_2} >\tfrac{a+1}{a-1}$ 
 & \ding{55} (Thm.~\ref{thm_three_static_vs_dynamic}) & \checkmark (Thm.~\ref{thm_three_static_vs_dynamic}) \\
\hline
Pair \eqref{eq_two_scalar_sys_gaussian} with $a_1=a_2=a,\, b_1=-b_2=b$ 
 & \ding{55} (Thm.~\ref{thm_impossible_costabilization}) & \ding{55} (Thm.~\ref{thm_impossible_costabilization}) \\
 \hline
$\exp(n)$-hard pair  from \cite{zeng2023hardness} 
 & \ding{55} (\cite{zeng2023hardness}) & \checkmark (Thm.~\ref{thm_static_vs_dynamic_zeng_exp}) \\
\hline
$\exp(n)$-hard pair  from \cite{tsiamis2022learning}  
 & \ding{55} (\cite{tsiamis2022learning}) & \ding{55} (Thm.~\ref{coro_tasos}) \\
\hline
\end{tabular}
\end{table}

In this section, we compare the co-stabilization capabilities of linear static state feedback and linear dynamic state feedback controllers on scalar systems and $\exp(n)$-hard systems. We show that linear dynamic state feedback controllers are strictly more expressive in certain settings than linear static state feedback controllers.

\subsection{Scalar Systems}

We first consider a class of discrete-time scalar systems. Let
\begin{equation}
\label{eq_two_scalar_sys_gaussian}
\begin{aligned}
S_i &: \quad x_{t+1} = a_i x_t + b_i u_t,
\end{aligned}
\end{equation}
where $i\in\{1,2\}$, $x_t, u_t \in \mathbb{R}$, and $b_i \neq 0$ for $i=1,2$.

The following theorem shows that linear dynamic state feedback can strictly enlarge the class of co-stabilizable system pairs compared to linear static state feedback.

\begin{theorem}
\label{thm_three_static_vs_dynamic}
The following two scalar pairs cannot be co-stabilized by any linear static state feedback controller, but can be co-stabilized by a linear dynamic state feedback controller:
\begin{itemize}
    \item Systems in \eqref{eq_two_scalar_sys_gaussian} with $|a_1|,|a_2|>1$,  $|a_1 - a_2| > 2$, and $b_1 = b_2$;
    \item Systems in \eqref{eq_two_scalar_sys_gaussian} with $a_1 = a_2 = a > 1$ and 
   $
    \frac{b_1}{b_2} > \frac{a+1}{a-1}.
    $
\end{itemize}
\end{theorem}

Note that this result is existential and not constructive. In general, we do not know how much memory the dynamic state feedback controller would require. 
On the other hand, for some fixed small values of the memory, we can further quantify the advantage of dynamic controllers by characterizing the co-stabilization gap, i.e., the maximum distance in the parameter space that still allows co-stabilization.

\begin{proposition}
\label{lem_co_stab_gao_mem_2}
Consider $S_1$ and $S_2$ in \eqref{eq_two_scalar_sys_gaussian} with $a_1 \neq a_2$, $|a_1|,|a_2|>1$, and $b_1=b_2=b\neq 0$. Then:
\begin{itemize}
\item a co-stabilizing linear static state feedback controller exists if and only if $|a_1-a_2|<2$;
\item a co-stabilizing linear state-history feedback controller of the form $
u_t=k_0 x_t + k_1 x_{t-1}$
exists if and only if $|a_1-a_2|<4$.
\end{itemize}
\end{proposition}

Proposition~\ref{lem_co_stab_gao_mem_2} shows that even a finite-memory dynamic controller significantly enlarges the admissible co-stabilization region compared to static feedback.
However, as Proposition~\ref{lem_co_stab_gao_mem_2} illustrates, finite controller memory may still lead to a bounded co-stabilization gap.
The general quantitative relationship between controller memory and achievable co-stabilization gap remains open.

We next present a pair of systems that cannot be co-stabilized even by linear dynamic state feedback.

\begin{theorem}
\label{thm_impossible_costabilization}
Consider $S_1$ and $S_2$ in \eqref{eq_two_scalar_sys_gaussian} with $a_1=a_2=a$, $|a|\geq 1$, and $b_1 = -b_2 = b > 0$. Then no linear static or linear dynamic state feedback controller can co-stabilize this pair.
\end{theorem}

This result shows that dynamic state feedback does not universally overcome structural obstructions to co-stabilization.

\subsection{$\exp(n)$-Hard Systems}

This naturally raises the question of whether dynamic state feedback can remove the $\exp(n)$ hardness of learning to stabilize, which was established for static state feedback in \cite{zeng2023hardness,tsiamis2022learning}.
Consider the parametrized pair:
\begin{equation}
\label{eq_zeng}
{\small
\begin{aligned}
S_i:\quad
\mathbf{A} &=
\begin{bmatrix}
r & \alpha_i v & 0 & \cdots & 0 \\
0 & 0 & v & \cdots & 0 \\
& & \ddots & \ddots & \\
0 & 0 & 0 & \cdots & v \\
0 & 0 & 0 & \cdots & 0
\end{bmatrix},
\quad
\mathbf{B}_i =
\begin{bmatrix}
b^{(i)} \\
0 \\
\vdots \\
0 \\
v
\end{bmatrix},
\end{aligned}}
\end{equation}
where $i\in\{3,4\}$ and $n\geq 2$.

In the construction of $\exp(n)$-hard systems in \cite{zeng2023hardness}, co-stabilization by a linear static state feedback controller requires the system parameters to be exponentially close in the state dimension $n$.
\begin{theorem}
\label{thm_static_vs_dynamic_zeng_exp}
For the pair in \eqref{eq_zeng} with $r>1$, $\alpha_3=\alpha_4=1$, $0<v<\frac{r-1}{2}$, $b^{(3)}=0$, $b^{(4)}=\bar{b}$, and
$
\bar{b} > \left(\frac{2v}{r-1}\right)^n$, 
no linear static state feedback controller can co-stabilize the systems, whereas a linear dynamic state feedback controller can.
\end{theorem}

Theorem~\ref{thm_three_static_vs_dynamic}, together with Theorem~\ref{thm_static_vs_dynamic_zeng_exp}, demonstrates that linear dynamic state feedback is strictly more expressive than linear static state feedback in the co-stabilization problem. However, this phenomenon does not extend to all $\exp(n)$-hard instances. In particular, combining Theorem~\ref{thm_impossible_costabilization} with the construction of $\exp(n)$-hard systems in \cite{tsiamis2022learning}, we obtain the following.

\begin{theorem}
\label{coro_tasos}
For the pair in \eqref{eq_zeng} with $r>1$, $\alpha_3=1$, $\alpha_4=-1$, $0<v<1$, and $b^{(3)}=b^{(4)}=0$, no linear dynamic state feedback controller can co-stabilize the systems.
\end{theorem}

Therefore, linear dynamic state feedback might not eliminate the $\exp(n)$ hardness of learning to stabilize established in \cite{tsiamis2022learning,zeng2023hardness}. 


\section{Co-Stabilization by Path Integral with State-History Feedback}
\label{sec_algorithm_path_integral}

In the previous section, we established that linear dynamic state feedback is strictly more expressive than static feedback for co-stabilization. We now turn to the algorithmic problem of computing such controllers.

Finding a co-stabilizing controller for a given set of systems is, in general, hard  {\cite{blondel1993simultaneous}}. Existing work based on linear matrix inequalities (LMIs) provides sufficient conditions by requiring a shared Lyapunov function (see, e.g., Section 7.2.3 in \cite{blanchini2008set}). Policy gradient methods have recently emerged as an alternative to LMIs, sometimes with global or local convergence guarantees \cite{hu2023toward, fujinami2025policy}. In this section, inspired by the path-integral policy search framework developed in \cite{williams2017model, yi2024covo}, and the stabilization strategy proposed in \cite{zhao2024convergence}, we propose a path integral algorithm for computing co-stabilizing controllers for a finite set of linear systems. 

We first present our algorithm for linear static state feedback controllers. Then we will show that based on the reparametrizations,  the same algorithm can be used to search for dynamic controllers. 
{\begin{remark}
    The previous section focuses on the theoretical characterization of co-stabilization of two systems, reflecting the inherent limitations of existing analytical tools in robust control \cite{blondel1993simultaneous}. In contrast, the algorithm developed here is applicable to co-stabilize multiple systems.
\end{remark}}
Consider a finite collection $\left\{
(\mathbf{A}_i,\mathbf{B}_i)
\right\}_{i\in[M]}$ of linear systems. For a fixed feedback gain $\mathbf{K}$ and discount factor $\gamma\in(0,1)$,
define the infinite-horizon discounted federated LQR cost
$
J_{\mathrm{Co}}(\gamma,\mathbf{K})
:=
\frac{1}{M}
\sum_{i=1}^{M}
J(\gamma,\mathbf{K},\mathbf{A}_i,\mathbf{B}_i),
$
where
\begin{equation}
\label{eq:discounted_LQR_clean}
\begin{aligned}
J(\gamma,\mathbf{K},\mathbf{A}_i,\mathbf{B}_i)
:=
\mathbb{E}_{\mathbf{x}_0}
\left[
\sum_{t=0}^{\infty}
\left(
\mathbf{x}_t^\top \Q \mathbf{x}_t
+
\mathbf{u}_t^\top \mathbf{R}\mathbf{u}_t
\right)
\right],
\end{aligned}
\end{equation}
subject to
\begin{equation}
\label{discounted_lqr_sys}
\mathbf{u}_t = \mathbf{K}\mathbf{x}_t,
\qquad
\mathbf{x}_{t+1}
=
\sqrt{\gamma}
(\mathbf{A}_i+\mathbf{B}_i\mathbf{K})
\mathbf{x}_t,
\end{equation}
with $\mathbb{E}[\mathbf{x}_0\mathbf{x}_0^\top]=\mathbf{I}_n$,
$\Q\succeq 0$, and $\mathbf{R}\succ 0$.
For later use, define
\begin{equation}
\overline{J}(\gamma,\mathbf{K})
:=
\max_{i\in[M]}
J(\gamma,\mathbf{K},\mathbf{A}_i,\mathbf{B}_i).
\end{equation}
For each system and fixed gain $\mathbf{K}$,
define $\mathbf{P}^{i,\gamma}_{\mathbf{K}}$
and $\boldsymbol{\Sigma}^{i,\gamma}_{\mathbf{K}}$
as the unique positive semidefinite solutions of the discrete Lyapunov equations
\begin{equation}
\label{eq:def_P_clean}
\mathbf{P}^{i,\gamma}_{\mathbf{K}}
=
\Q
+
\mathbf{K}^\top\mathbf{R}\mathbf{K}
+
\gamma
(\mathbf{A}_i+\mathbf{B}_i\mathbf{K})^\top
\mathbf{P}^{i,\gamma}_{\mathbf{K}}
(\mathbf{A}_i+\mathbf{B}_i\mathbf{K}),
\end{equation}
and
\begin{equation}
\boldsymbol{\Sigma}^{i,\gamma}_{\mathbf{K}}
=
\mathbf{I}_n
+
\gamma
(\mathbf{A}_i+\mathbf{B}_i\mathbf{K})
\boldsymbol{\Sigma}^{i,\gamma}_{\mathbf{K}}
(\mathbf{A}_i+\mathbf{B}_i\mathbf{K})^\top.
\end{equation}

Under the assumption
$\mathbb{E}[\mathbf{x}_0\mathbf{x}_0^\top]=\mathbf{I}_n$,
Lemma~1 of \cite{fazel2018global} implies
$
J(\gamma,\mathbf{K},\mathbf{A}_i,\mathbf{B}_i)
=
\operatorname{trace}
\left(
\mathbf{P}^{i,\gamma}_{\mathbf{K}}
\right).
$

Hence, the federated objective in \eqref{eq:discounted_LQR_clean} can be evaluated by solving
$M$ Lyapunov equations.

We now present the co-stabilization path-integral algorithm,
summarized in Alg.~\ref{alg:CS_PI_mppi_update_simplified}, where $\underline{\sigma}$ denotes the least singular value of a matrix.
The inner loop performs a zeroth-order stochastic descent
via exponential reweighting of sampled perturbations, which follows the path integral idea from \cite{williams2017model, yi2024covo}. The update of $\alpha_j$ follows the stabilization idea in \cite{zhao2024convergence}.
Unstable candidate gains are penalized by assigning
a large cost $J_{\max}$.

\begin{algorithm}[htbp]
\caption{Co-Stabilization by Path Integral}
\label{alg:CS_PI_mppi_update_simplified}
\begin{algorithmic}
\Require
Systems $\{(\mathbf{A}_i,\mathbf{B}_i)\}_{i=1}^M$; 
$\mathbf{Q}\succeq 0$, $\mathbf{R}\succ 0$; 
initial gain $\mathbf{K}_0$; 
$\gamma_0\in(0,1)$; 
$\lambda>0$; 
$r>0$; 
$T$; $N$; 
$\xi\in(0,1)$; $\underline{\alpha}$;
$J_{\max}>0$.
\Ensure
Returns the feasibility flag and gain.

\State Initialize $\mathbf{K}_0$, $\gamma_0$.
\For{$j=0,1,2,\dots$}
    \State $\widehat{\mathbf{K}}^{(j)}_0 \gets \mathbf{K}_j$
    \For{$t=0,\dots,T-1$}
        \For{$n=1,\dots,N$}
            \State Uniformly sample $\Delta_{\mathbf{K},n}$ s.t. $\|\Delta_{\mathbf{K},n}\|_F=r$
            \State $\mathbf{K}_n \gets \widehat{\mathbf{K}}^{(j)}_t + \Delta_{\mathbf{K},n}$
            \If{$\rho(\sqrt{\gamma_j}(\mathbf{A}_i+\mathbf{B}_i\mathbf{K}_n))\ge 1$ for some $i$}
                \State $J_n \gets J_{\max}$
            \Else
                \State $J_n \gets J_{\mathrm{Co}}(\gamma_j,\mathbf{K}_n)$
            \EndIf
        \EndFor
        \State $J_{\min}\gets\min_n J_n$
        \State $w_n\gets\exp\!\left(-\frac{J_n-J_{\min}}{\lambda}\right)$, $v_n\gets \frac{w_n}{\sum_k w_k}$
        \State
        $
        \widehat{\mathbf{K}}^{(j)}_{t+1}
        \gets
        \widehat{\mathbf{K}}^{(j)}_t
        +
        \sum_{n=1}^N
        v_n
        \Delta_{\mathbf{K},n}
        $
    \EndFor
    \State $\mathbf{K}_{j+1}\gets\widehat{\mathbf{K}}^{(j)}_T$
    \State
    $
    \alpha_j
    \gets
    \frac{
    \underline{\sigma}(\mathbf{Q}+\mathbf{K}_{j+1}^\top\mathbf{R}\mathbf{K}_{j+1})
    }{
    \overline{J}(\gamma_j,\mathbf{K}_{j+1})
    -
    \underline{\sigma}(\mathbf{Q}+\mathbf{K}_{j+1}^\top\mathbf{R}\mathbf{K}_{j+1}) 
    }
    $
    \State
    $
    \gamma_{j+1}
    \gets
    (1+\xi\alpha_j)\gamma_j
    $
    \If{$\underset{i\in [M]}{\max}\rho(\mathbf{A}_i+\mathbf{B}_i\mathbf{K}_{j+1})<1$ }
        \State \Return $(\mathrm{True},\mathbf{K}_{j+1})$
    \EndIf
     \If{$\gamma_{j+1}\ge 1$ \textbf{or} $\alpha_{j}\le \underline{\alpha}$ }
        \State \Return $(\mathrm{False},\mathbf{K}_{j+1})$
    \EndIf
\EndFor
\end{algorithmic}
\end{algorithm}

Next, we show how to use Alg.~\ref{alg:CS_PI_mppi_update_simplified} to get a co-stabilizing linear state-history feedback controller.
For any $h\ge 1$, define the history-augmented state
\begin{equation}
\mathbf{z}_t=
\begin{bmatrix}
\mathbf{x}_t^\top\; \mathbf{x}_{t-1}^\top\; \cdots\; \mathbf{x}_{t-h+1}^\top
\end{bmatrix}^\top\in\mathbb{R}^{hn}.
\end{equation}
Then we have the lifted dynamics
\begin{equation}
\mathbf{z}_{t+1}=\mathbf{A}_{i}^{hi}\mathbf{z}_t+\mathbf{B}_{i}^{hi}\mathbf{u}_t, \qquad i\in [M],
\end{equation}
where
\begin{equation}
\label{def_A_B_his}
\mathbf{A}_{i}^{hi}=
\begin{bmatrix}
\mathbf{A}_i & 0 & \cdots & 0\\
\mathbf{I}_n & 0 & \cdots & 0\\
0 & \mathbf{I}_n & \ddots & \vdots\\
\vdots & & \ddots & 0\\
0 & \cdots & \mathbf{I}_n & 0
\end{bmatrix},
\qquad
\mathbf{B}_{i}^{hi}=
\begin{bmatrix}
\mathbf{B}_i\\ 0\\ \vdots\\ 0
\end{bmatrix}.
\end{equation}

Then, by the reparametrizations for linear  state-history feedback controllers in \eqref{def_A_B_his}, co-stabilization of linear state-history feedback reduces to
linear static state feedback co-stabilization of the augmented systems,
to which Alg.~\ref{alg:CS_PI_mppi_update_simplified} applies directly.

\begin{remark}[Limits of LMI-based Co-Stabilization with Memory]
Although linear dynamic state feedback and state-history feedback can strictly enlarge the set of co-stabilizable systems, this advantage is not captured by standard LMI-based approaches based on a common quadratic Lyapunov function in \cite{blanchini2008set}. In particular, it can be proved that augmenting the system with linear dynamic state feedback or state-history feedback does not enlarge the feasibility region of the corresponding co-stabilization LMI. 
This highlights a fundamental limitation of LMI-based methods: while dynamic controllers provide additional expressive power for co-stabilization, convex formulations based on common Lyapunov functions fail to exploit this benefit.
\end{remark}

 \section{Experiments}
 \label{sec_exp}

We conduct three numerical experiments to validate the theoretical results developed in the previous sections. 
The first two experiments are to check how Alg.~\ref{alg:CS_PI_mppi_update_simplified} with linear state-history feedback\footnote{Empirically, we have not observed a substantial performance difference when using linear dynamic feedback versus linear state-history feedback with our algorithm. Hence, we restrict our experiments to the latter.} is affected the controller memory and the original state dimension. The third experiment is to compare Alg.~\ref{alg:CS_PI_mppi_update_simplified} with a co-stabilization policy gradient method from the literature.


\subsection{Horizon vs.\ Co-Stabilization Gap}
\label{subsec_mem_cos_exp}
We first investigate how increasing the horizon of a linear state-history controller enlarges the co-stabilization region.
Consider the scalar pair in \eqref{eq_two_scalar_sys_gaussian} with
$a_1=a_2=1.2$, $b_1=1$, and $b_2>0$.
We apply Alg.~\ref{alg:CS_PI_mppi_update_simplified} to compute a co-stabilizing linear state-history feedback controller for $\{(a_i,b_i)\}_{i=1,2}$.
For horizon $h$, the systems are lifted according to \eqref{def_A_B_his}.

The hyperparameters of Alg.~\ref{alg:CS_PI_mppi_update_simplified} are: $\K_0=\mathbf{0}$,
$
\mathbf{Q}^{\mathrm{his}}=\mathbf{I}_h,\quad
\mathbf{R}^{\mathrm{his}}=1$, $
\xi=0.99$, $
\lambda=10^{-4},\quad
r=1\times 10^{-1}$,
$
T=20$, 
$N=20$, 
$J_{\max}=10^{12}$,
and
$
\gamma_0
=
 1/{\left(1+\max\{1,\max_{i\in\{1,2\}}
\rho(\mathbf{A}^{\mathrm{his}}_i+\mathbf{B}^{\mathrm{his}}_i\widehat{\mathbf{K}}^{\mathrm{his}}_0)\}^2
\right)
}.
$

\paragraph{Feasibility criterion}
We declare Alg.~\ref{alg:CS_PI_mppi_update_simplified} feasible for $\{(\mathbf{A}_i,\mathbf{B}_i)\}_{i=1,2}$ if, upon termination (i.e., $\gamma_{j+1}\ge 1$ or $\alpha_j\le \underline{\alpha} = 10^{-6}$), the returned gain $\widehat{\mathbf{K}}^{\mathrm{his}}_{j+1}$ satisfies
\[
\rho(\mathbf{A}^{\mathrm{his}}_i+\mathbf{B}^{\mathrm{his}}_i\widehat{\mathbf{K}}^{\mathrm{his}}_{j+1})<1
\quad
\forall i\in\{1,2\}.
\]
Otherwise, the algorithm is declared infeasible.

\paragraph{Largest admissible gap}
For each horizon $h$, we apply a bisection procedure to determine the largest $\bar{b}_2$ such that Alg.~\ref{alg:CS_PI_mppi_update_simplified} remains feasible.
Fig.~\ref{fig:largest_b2_vs_h} plots  $\bar{b}_2$ as a function of $h$.

The results in Fig.~\ref{fig:largest_b2_vs_h} show that in our experiments, the estimated feasible region increases monotonically with horizon. 
In particular, state-history feedback ($h>1$) achieves a strictly larger $\bar{b}_2$ than static state feedback ($h=1$), empirically supporting Theorem~\ref{thm_three_static_vs_dynamic}.

\begin{figure}[t]
  \centering
  \includegraphics[width=0.8\columnwidth]{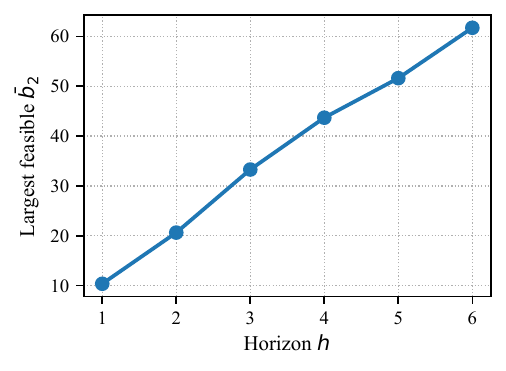}
  \caption{Largest feasible $b_2$ versus horizon $h$. In addition, the largest feasible $b_2$ by solving LMI with a common quadratic Lyapunov function is $10.99$. }
  \label{fig:largest_b2_vs_h}
  \vspace{-0.4cm}
\end{figure}

\subsection{Exponential Hardness in System Dimension}

We next examine how the co-stabilization gap scales with system dimension.
Consider the system in \eqref{eq_zeng} with parameters
$r=1.2$, $v=0.5$, $b^{(3)}=0$, and $b^{(4)}=\bar{b}>0$.
We vary the state dimension $n\in\{2,3,4,5\}$.

We compare three approaches: i)
 LMI-based common quadratic Lyapunov method, which is also employed in \cite{zeng2023hardness}, ii)
 linear static state-feedback co-stabilization using Alg.~\ref{alg:CS_PI_mppi_update_simplified} ($h=1$); and iii)
  linear state-history feedback co-stabilization using Alg.~\ref{alg:CS_PI_mppi_update_simplified} ($h=3$). The hyperparameters of Alg.~\ref{alg:CS_PI_mppi_update_simplified} follow the setup in Section \ref{subsec_mem_cos_exp}.
 
Fig.~\ref{fig:largest_b2_vs_n} shows $\bar{m}$ versus $n$.
While the chosen $v$ does not satisfy the conditions in Theorem~\ref{thm_static_vs_dynamic_zeng_exp}, all methods exhibit exponential decay of the feasible co-stabilization gap as $n$ increases, suggesting fixed-horizon linear state-history feedback may not remove the $\exp(n)$ hardness identified in \cite{zeng2023hardness}.

\begin{figure}[t]
  \centering
  \includegraphics[width=0.8\columnwidth]{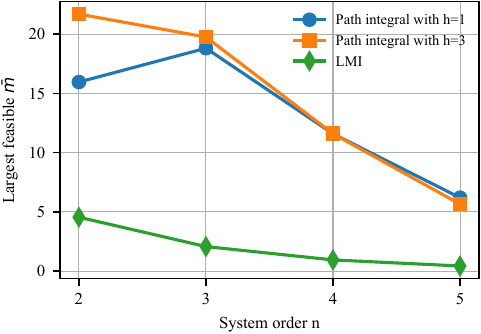}
  \caption{Largest feasible $\bar{b}$ versus system dimension $n$.}
  \label{fig:largest_b2_vs_n}
  \vspace{-0.2cm}
\end{figure}

\subsection{Path Integral vs.\ Policy Gradient}
 In this part, we compare our Alg. \ref{alg:CS_PI_mppi_update_simplified} with the co-stabilization policy gradient algorithm given in \cite{fujinami2025policy} for linear static state feedback ($h=1$) and linear state-history feedback with $h=3$. We implement the comparison experiment on the co-stabilization of two scalar systems and two single-input order-$2$ systems. The first one is the scalar pair in \eqref{eq_two_scalar_sys_gaussian} with
$a_1=a_2=1.2$ and $b_1,b_2$ are uniformly sampled from $[-1,11]$.  The second one is the discretized and
linearized inverted pendulum system:
\begin{equation}
\label{eq_linearized_in_pen}
\A=\left[\begin{array}{cc}
1 & d t \\
\frac{g}{\ell} d t & 1
\end{array}\right],\; \B=\left[\begin{array}{c}
0 \\
\frac{d t}{m \ell^2}
\end{array}\right],
\end{equation}
with known $dt=1e-4$ and $g=10$. For unknown $l$ and $m$, we independently sample two values from uniform distributions. Specifically, for each sample $i = 1,2$, we draw $m_i \sim \mathcal{U}[0.75, 1.25]$ and $\ell_i \sim \mathcal{U}[0.75, 1.25]$, and define the corresponding linearized inverted pendulum system $(\A_i, \B_i)$ by \eqref{eq_linearized_in_pen}. The hyperparameters of Alg.~\ref{alg:CS_PI_mppi_update_simplified} follow the setup in Section \ref{subsec_mem_cos_exp}. The co-stabilization policy gradient algorithm uses the original implementation in the GitHub repository of \cite{fujinami2025policy}. For each case, we run $20$ experiments independently to compute the co-stabilization success rates, which are summarized in Table \ref{tab:comparison}. Based on the table, we can see that the path-integral method in Alg. \ref{alg:CS_PI_mppi_update_simplified} is better than the co-stabilization policy gradient, and the controller memory increases the performance of Alg. \ref{alg:CS_PI_mppi_update_simplified} for the two implemented cases. A limitation of our Alg. \ref{alg:CS_PI_mppi_update_simplified} is that it is more time-expensive than the co-stabilization policy gradient.

 \begin{table}[t]
\centering
\caption{Success rates comparison.}
\vspace{-0.3cm}
\label{tab:comparison}
\setlength{\tabcolsep}{4pt}
\begin{tabular}{lcccc}
\toprule
 & \makecell{{\cite{fujinami2025policy}} \\ ($h=1$)} 
 & \makecell{{Alg. \ref{alg:CS_PI_mppi_update_simplified}} \\ ($h=1$)} 
 & \makecell{{\cite{fujinami2025policy}} \\ ($h=3$)} 
 & \makecell{{Alg. \ref{alg:CS_PI_mppi_update_simplified}} \\ ($h=3$)} \\
\midrule
Scalar & 0.15 & {\bf 0.80} & 0.05  & {\bf 0.90} \\
Inv.\ Pend. & 0.20  & {\bf 1.00} & 0.55  & {\bf 1.00} \\
\bottomrule
\end{tabular}
\vspace{-0.5cm}
\end{table}

\section{Conclusion and Future Work}
\label{sec_con}
In this paper, we study the role of controller memory in co-stabilization. We show that linear dynamic state feedback strictly enlarges the class of co-stabilizable systems compared to static feedback, while fundamental limitations remain.
On the algorithmic side, we develop a path-integral method for computing co-stabilizing controllers.
Future work includes understanding how the controller memory affects the sample complexity of learning-to-stabilize problems. We are also interested in investigating the convergence properties of the proposed path integral algorithm for co-stabilization.

\vspace{0.2cm}
\noindent {\bf Acknowledgments: } This work is supported in part by ONR grant
N00014-21-1-2431 (CLEVR-AI). NO would like to thank Constantino Lagoa for some inspiring discussions on co-stabilization.


\appendix

\subsection{Some Frequency Domain Preliminary Results }
\label{ape_fre_domin_tool}
In this section, we introduce some continuous-time frequency domain tools from the book \cite{doyle2013feedback}.
First, we need to introduce coprime factorization, which plays a critical role in feedback and robust control, followed by two useful lemmas.

\begin{definition}[Coprime Factorization]
Consider a continuous-time transfer function $P^c(s)$. A set of four stable, proper\footnote{A continuous-time transfer function is stable if its poles are all in the left-half complex plane and is proper if the degree of the numerator does not exceed the degree of the denominator.} transfer functions $N(s)$, $M(s)$, $X(s)$ and $Y(s)$ is called a coprime factorization of $P^c(s)$ if
$P^c(s)= N(s) / M(s)$ and
$
N(s) X(s)+M(s) Y(s)=1,
$.
\end{definition}

\begin{lemma}[{\cite[Theorem 3 of Chapter 5]{doyle2013feedback}}]
\label{lem_iff_strong_stab}
   A continuous-time transfer function $P^c(s)$ is strongly stabilizable if and only if it has an even number of real poles between every pair of non-negative real zeros.  
\end{lemma}

\begin{lemma}[{\cite[Theorem 4 of Chapter 5]{doyle2013feedback}}]
\label{lem_iff_co_stabilizable}
Consider the continuous-time transfer functions $P^c_1(s)$ and $P^c_2(s)$. 
 There exists a proper continuous-time transfer function $C^c(s)$ such that both $\frac{1}{1+C^c(s)P^c_1(s)}$ and $\frac{1}{1+C^c(s)P^c_2(s)}$ are stable transfer functions if and only if there exist coprime factorizations $(N_1(s), M_1(s), X_1(s), Y_1(s))$ and  $(N_2(s), M_2(s), X_2(s), Y_2(s))$ of $P^c_1(s)$ and $P^c_2(s)$, respectively, such that
\begin{equation}
\label{def_hat_P}
\hat{P}(s) :=\frac{N_2(s) M_1(s) -N_1(s) M_2(s)}{N_2(s) X_1(s)+M_2(s) Y_1(s)},
\end{equation}
 is strongly stabilizable.
 
 
\end{lemma}

Consider two fully-observed LTI systems $(\A_1,\B_1)$ and $(\A_2,\B_2)$ defined in \eqref{LTI}, for which the output matrix is $\I_n$. The open-loop discrete-time transfer functions from input to state for these two systems are  $
    P^d_i(z)   = (z\I_n - \A_i)^{-1}\B_i.
$
The discrete-time transfer function from state to input  for the linear dynamic state feedback controller defined in \eqref{p_dynamic_state_feedback} is 
$C^d(z)=\K + \Hb (z \I_p - \F)^{-1} \G.$ 
{In particular, in the following analysis, we restrict attention to scalar (SISO)
transfer functions obtained from the systems under consideration via appropriate reductions.}
With the bilinear transformation $z=\frac{s+1}{s-1}$, we can transfer the discrete-time transfer functions into continuous-time transfer functions $P^c_i(s) = P^d_i\left(\frac{s+1}{s-1}\right)$ with $i=1,2$ and $C^c(s) = C^d\left(\frac{s+1}{s-1}\right)$. Then, the co-stabilization of $(\A_1,\B_1)$ and $(\A_2,\B_2)$ by linear dynamic state feedback controller defined in \eqref{p_dynamic_state_feedback} is equivalent to finding a proper $C^c(s)$ such that both $\frac{1}{1+C^c(s)P^c_1(s)}$ and $\frac{1}{1+C^c(s)P^c_2(s)}$ are stable transfer functions (More details can be found in Chapter 5 of \cite{doyle2013feedback}.). This allows us to use Lemma \ref{lem_iff_co_stabilizable} to prove our main results.

\subsection{Proof of Theorem \ref{thm_three_static_vs_dynamic}}
\label{proof_thm_three_static_vs_dynamic}

\begin{proof}[Proof of Theorem \ref{thm_three_static_vs_dynamic}]  
\noindent \textbf{Case 1:  $S_1$ and $S_2$ in \eqref{eq_two_scalar_sys_gaussian} with $|a_1|,|a_2|>1$, $|a_1 - a_2| > 2$, $b_1 = b_2$.}

The proof for the linear static state feedback controller follows easily by checking the sign of $a_1+b_1 k$ and $a_2+b_2 k$. Next, we prove the result for the linear dynamic state feedback controller.
The discrete-time  transfer functions of $S_1$ and $S_2$ in \eqref{eq_two_scalar_sys_gaussian} are
$
  P^d_i(z) = \frac{b_i}{z - a_i},
$
with $i=1,2$. Mapping the discrete-time transfer functions $P^d_1(z)$ and $P^d_2(z)$ into continuous-time transfer functions via the {one-to-one} bilinear transformation, we get
$
  P^c_i(s) = \frac{b_i(s-1)}{(1 - a_i)s+a_i+1},   
$
with $|a_i|\neq 1$ and $b_i \neq  0$, for $i=1,2$.
We can then construct a coprime factorization of $P_i(s)$ for $i=1,2$ as follows:
\begin{equation}
\label{eq_coprime_P12}
    \begin{aligned}
N_i(s)   &=\frac{b_i(s-1)}{s+1}  ,
M_i(s)   =\frac{(1-a_i) s+(a_i+1)}{s+1}  , \\
  X_i(s)   & =\frac{ \frac{2 a_i-1}{b_i} s-\frac{1}{b_i}}{s+1}  ,
Y_i(s)   =\frac{2 s}{s+1}.
    \end{aligned}
\end{equation}
Based on \eqref{def_hat_P} and \eqref{eq_coprime_P12}, noting that $b_1=b_2\neq 0$, we obtain:
\begin{equation}
{\small
\begin{aligned}
\hat{P}(s) 
&=  \frac{  b_1(a_2-a_1)(s-1)^2}{(s-1)[(2 a_1-1) s-1]+2 s[(1-a_2) s+(1+a_2)]}.
\end{aligned}}
\end{equation}
Next, we need to determine if the above continuous-time transfer function $ \hat{P}(s) $ is strongly stabilizable. The above continuous-time transfer function $ \hat{P}(s) $ has two repeated unstable real zeros at $1$, and $1$ is not a pole of $ \hat{P}(s) $. Therefore, $ \hat{P}(s) $ does not have real poles between its unstable zeros, which, based on Lemma \ref{lem_iff_strong_stab}, implies that $ \hat{P}(s)$ is strongly stabilizable for all  $b_1=b_2\neq 0$. Then by Lemma \ref{lem_iff_co_stabilizable}, we complete the proof.

\noindent \textbf{Case 2:  $S_1$ and $S_2$ in \eqref{eq_two_scalar_sys_gaussian} with $a_1 = a_2=a>1,$  $ \frac{b_1}{b_2}  >  \frac{a+1}{a-1}$.}

As before, the proof for the linear static state feedback controller can be done by checking the stability of $a_1+b_1 k$ and $a_2+b_2 k$. Next, we focus on the proof for linear dynamic state feedback controllers. 
Based on \eqref{def_hat_P} and \eqref{eq_coprime_P12}, for $a_1=a_2=a>1$ and $\frac{b_1}{b_2} >  \frac{a+1}{a-1}$, we have
\begin{equation}
\begin{aligned}
& \hat{P}(s)  
= \\
& \frac{\left(b_2-b_1\right)(s-1)((1-a) s+(a+1))}{\left[ \frac{b_2}{b_1}(2a - 1) + 2(1 - a) \right] s^2 
+ \left[  2a\left( 1-\frac{b_2}{b_1}\right) + 2 \right] s 
+ \frac{b_2}{b_1}}.
\end{aligned}
\end{equation}
Next, we need to determine if the above continuous-time transfer function $ \hat{P}(s)$ is strongly stabilizable. Firstly, the above continuous-time transfer function $ \hat{P}(s)$ has two independent unstable real zeros at $1$ and $\frac{a+1}{a-1}$ respectively. Noting that the denominator of $ \hat{P}(s)$ is a quadratic and it is greater than zero when evaluated at $1$ and $\frac{a+1}{a-1}$, we conclude that there are either no poles or two real poles between the zeros. Hence, invoking Lemmas \ref{lem_iff_strong_stab} and \ref{lem_iff_co_stabilizable} concludes the proof.
\end{proof}

\subsection{Proof of Proposition \ref{lem_co_stab_gao_mem_2}}
\label{proof_lem_co_stab_gao_mem_2}
Before proving Proposition \ref{lem_co_stab_gao_mem_2}, we present the Jury test for discrete-time stability of order-$2$ polynomials.
\begin{lemma}[Jury stability test for order-$2$ polynomial, Theorem 4.6 in \cite{fadali2012digital}]
\label{lem_jury_test}
Consider an order-$2$ polynomial
$
 z^2 + q_{1} z +   q_0 
$. All roots of this polynomial are inside the unit circle 
    if and only if $
  -(1   + q_0)< q_1 < 1   + q_0     
  $ and $  |q_0| < 1
$.
 \end{lemma}
 
\begin{proof}[Proof of Proposition \ref{lem_co_stab_gao_mem_2}]
The case of linear static state feedback controller follows as before from checking the stability of $a_1+b_1 k$ and $a_2+b_2 k$ with $k\in \Rb$.
For the memory-2 linear state-history feedback, the augmented closed-loops for $i=1,2$ take the form 
\begin{equation}
    \begin{bmatrix} x_{t+1} \\ x_{t} \end{bmatrix} = \underbrace{\begin{bmatrix} a_i + b_i k_0 & b_i k_1 \\ 1 & 0 \end{bmatrix}}_{=: \D_{a_i,b_i}} \begin{bmatrix} x_t \\ x_{t-1} \end{bmatrix},
\end{equation}
The characteristic polynomial of $\D_{a_i,b_i}$ is given by:
\begin{equation}
\label{eq_kfg}
    \Delta_{i}(z) := \det (z \I_2 - \D_{a_i,b_i})= z^2 - (a_i + b_i k_0) z - b_i k_1.
\end{equation}
Let $k_0, k_1$ stabilize $\D_{a_1,b_1}$ so that the eigenvalues of $\D_{a_1,b_1}$ satisfy $|p_1^{cl}|<1$ and $|p_2^{cl}|<1$.
Then,
\begin{equation}
\label{eq_k12}
     \Delta_{1}(z) = (z - p_1^{cl})(z - p_2^{cl})=z^2-(p_1^{cl}+p_2^{cl})z + p_1^{cl} p_2^{cl} .
\end{equation}
Based on \eqref{eq_kfg} and \eqref{eq_k12}, we can use $p_{1}^{cl}$ and $p_{2}^{cl}$ to express all stabilizing $(k_0,k_1)$ for $\D_{a_1,b_1}$,
\begin{equation}
\label{eq_form_k12}
    k_0   =  \frac{p_1^{cl}+p_2^{cl} - a_1}{b_1}\; \text{and} \; k_1  = - \frac{p_1^{cl}p_2^{cl} }{b_1}.
\end{equation}
Substituting $k_0$ and $k_1$ from \eqref{eq_form_k12} into $ \Delta_{2}(z)$, we have
\begin{equation}
\label{eq_solve_k12}
     \Delta_{2}(z)  =z^2-(a_2-a_1+p_1^{cl}+p_2^{cl})z + p_1^{cl} p_2^{cl} .
\end{equation}
 
Based on the conditions of Jury stability test in Lemma \ref{lem_jury_test}, $q_0$
 of  $\Delta_{i}(z) $ are the same for $i=1,2$. We just need to check the following condition on $q_1$ for $\Delta_{2}(z) $:
     \begin{equation}
     \label{eq:jr1d}
     \begin{aligned}
& -(1+p_1^{cl}p_2^{cl}) =-(1+q_0)   <     q_1 = -(a_2-a_1+p_1^{cl}+p_2^{cl} ) \\
& < 1+ q_0=1+p_1^{cl}p_2^{cl}.
\end{aligned}
     \end{equation}
     The result follows by simplifying \eqref{eq:jr1d} as:
      $$-4<-(1 +p_1^{cl})(1+p_2^{cl})<a_2 - a_1  <(1 -p_1^{cl})(1-p_2^{cl}) < 4.$$
\end{proof}

\subsection{Proof of Theorem \ref{thm_impossible_costabilization}}
\label{proof_thm_impossible_costabilization}
\begin{proof}[Proof of Theorem \ref{thm_impossible_costabilization}] Consider $S_1$ and $S_2$ in \eqref{eq_two_scalar_sys_gaussian} with $a_1=a_2=a$, $|a|\geq 1$, and $b_1 = -b_2 = b > 0$. 
Based on \eqref{def_hat_P} and \eqref{eq_coprime_P12},  for this setting, we have
\begin{equation}
{\small
\hat{P}(s) = \frac{-2b(s-1)[(1-a) s+(a+1)]}{\left(3-4 a\right) s^2+\left(4 a+2\right)}}
\end{equation}
Because $\left(4 a+2\right)^2+4\left(3-4 a\right)=16 a^2+16=16\left(a^2+1\right)>0$, $\hat{P}(s)$ has two real poles.

{In addition, it can be checked that the strong stabilizability of $\hat{P}(s)$ is invariant under different coprime factorizations of $P^c_1(s)$ and  $P^c_2(s)$. Then, if there exists one group of coprime factorization of $P^c_1(s)$ and  $P^c_2(s)$ such that $\hat{P}(s)$ is not strongly stabilizable, then for all other coprime factorizations, $\hat{P}(s)$ is also not strongly stabilizable.}

Next, we discuss the cases $|a|=1$ and $|a|>1$ separately. When $a=1$, $\hat{P}(s)$ have two nonnegative real zeros at $1$ and $+\infty$ and the denominator is $-s^2+6s-1$. Since the denominator is a quadratic with a leading negative coefficient, evaluating to a positive number at $1$, there exists a real pole of the system between $1$ and $+\infty$. Combining with Lemma \ref{lem_iff_strong_stab}, $ \hat{P}(s)$ is not strongly stabilizable for $a=1$. With a similar reasoning, we can show $ \hat{P}(s)$ is not strongly stabilizable for $a=-1$ either. 


When $|a|>1$,  $ \hat{P}(s)$ has two independent unstable  real zeros at $1$ and $\frac{a+1}{a-1}$ respectively. Next, we determine the positions of poles of  $ \hat{P}(s)$. First, recall that the denominator evaluates to a positive number at $1$, and it can be shown that it evaluates to $\frac{-4a^2}{(a-1)^2}<0$ at $\frac{a+1}{a-1}$. Therefore, 
it can be concluded that $ \hat{P}(s)$ has a real pole less than $1$ and one between $1$ and $\frac{a_1+1}{a_1-1}$.
Combining with Lemma \ref{lem_iff_strong_stab}, we get that $ \hat{P}(s)$ is not strongly stabilizable for all $|a|>1$. Then by Lemma \ref{lem_iff_co_stabilizable}, there does not exist a co-stabilizing linear dynamic state feedback controller for this pair. As the linear static state feedback controller is a special case of the linear dynamic state feedback controller, we complete the proof.
\end{proof}

\subsection{Proof of Theorem \ref{thm_static_vs_dynamic_zeng_exp}}
\label{proof_thm_static_vs_dynamic_zeng_exp}
\begin{proof}[Proof of Theorem \ref{thm_static_vs_dynamic_zeng_exp}] 

The proof for linear static state feedback controllers is finished in Proposition 1 in \cite{zeng2023hardness}. Next, let us finish the proof for the linear dynamic state feedback controllers of this case.  For the pair in \eqref{eq_zeng} with $r>1$, $\alpha_3=\alpha_4=1$, $0<v<\frac{r-1}{2}$, $b^{(3)}=0$, $b^{(4)}=\bar{b}$,
 $S_3$ and $S_4$ in \eqref{eq_zeng} are able to be reduced to the following two scalar systems: $x_{t+1}^{(1)} = r x_{t}^{(1)} + v^n u_{t-n+1} $ and $x_{t+1}^{(1)} = r x_{t}^{(1)} + v^n  u_{t-n+1}+ \bar{b} u_{t}$. The discrete-time transfer functions for these two systems are:
 \begin{equation}
P^d_3(z)  = \frac{ v^n}{(z-r)z^{ n-1 }},\;\text{and} \;
P^d_4(z)  = \frac{v^n +\bar{b} z^{n-1}}{(z-r)z^{ n-1 }}.
\end{equation}
Similar to Theorem \ref{thm_three_static_vs_dynamic}, we replace $z$ with $ \frac{s+1}{s-1}$ in $P^d_3(z)$ and $P^d_4(z)$ to get their continuous-time transfer functions as follows:
\begin{equation}
\begin{aligned}
&P^c_3(s) =   \frac{v^n  (s - 1)^n }{ \left[(1 + r) - (r - 1)s\right](s + 1)^{n - 1} }, \\
&P^c_4(s)  = \frac{v^n (s - 1)^{n  } + \bar{b} (s + 1)^{n - 1}(s - 1) }{ \left[(1 + r) - (r - 1)s\right](s + 1)^{n - 1} }.
\end{aligned}
\end{equation}
For $P^c_3(s)$, a coprime factorization is the following:
\begin{equation}
\begin{aligned}
& N_3(s)=\frac{v^n(s-1)^n}{(s+1)^n},  
& M_3(s)=\frac{(1+r)-(r-1) s}{s+1},  \\
& X_3(s) = \frac{  a_0 }{s+1},  
& Y_3(s) = \frac{ \sum_{k=0}^n b_k s^k  }{(s+1)^n},
\end{aligned}
\end{equation}
where $a_0 = \frac{2\,r^{n+1}}{(r-1)\,v^n}$, $ b_n = -\frac{1}{r-1}$, and 
\begin{equation}
{\small
\begin{aligned}
&b_k = \\
&\frac{1}{1+r} \sum_{j=0}^{k} 
\left( \frac{r-1}{\,1+r\,} \right)^{k-j}
\left[ \binom{n+1}{j} - \frac{2\,r^{n+1}}{\,r-1\,} \binom{n}{j} (-1)^{n-j} \right],
\end{aligned}}
\end{equation}
for $ k = 0,1,\dots,n-1.$ 
For $P^c_4(s)$, a coprime factorization has the following $N_4(s)$ and $M_4(s)$
\begin{equation}
\begin{aligned}
& N_4(s)=\frac{ v^n(s - 1)^{n  } + \bar{b} (s + 1)^{n - 1}(s - 1)}{(s+1)^n},  \\
& M_4(s)=\frac{(1+r)-(r-1) s}{s+1}.
\end{aligned}
\end{equation}
Because $M_4(s)=M_3(s)$ and $N_4(s)=N_3+\bar{b} \frac{s-1}{s+1}$,   we have
$$
\begin{aligned}
N_4 M_3-N_3 M_4 & =\bar{b} \frac{s-1}{s+1} M_3\\
&=\bar{b} \frac{(1+r)-(r-1) s}{s+1} \frac{s-1}{s+1}, 
\end{aligned}
$$
$$
\begin{aligned}
N_4 X_3+M_4 Y_3 & =\left(N_3 X_3+M_3 Y_3\right)+\bar{b} \frac{s-1}{s+1} X_3\\
& =1+\bar{b} \frac{s-1}{s+1} \frac{ a_0}{s+1}.
\end{aligned}
$$
Then, combining with \eqref{def_hat_P}, we get
\begin{equation}
\begin{aligned}
    \hat{P}(s) 
    &=  \frac{\bar{b}((1+r)-(r-1) s) \frac{s-1}{(s+1)^2}}{1+\bar{b} \frac{2 r^{n+1}}{r-1} \frac{s-1}{(s+1)^2}}\\
    &= \frac{\bar{b}((1+r)-(r-1) s) (s-1)}{(s+1)^2+\bar{b} \frac{2 r^{n+1}}{(r-1)v^n}(s-1)  }.
    \end{aligned}
\end{equation}
Firstly, the above continuous-time transfer function $ \hat{P}(s)$ has two independent unstable real zeros at $1$ and $\frac{r+1}{r-1}$ respectively. Next, we determine the positions of the poles of  $ \hat{P}(s)$.
Consider the denominator
\begin{equation}
    Q(s)=(s+1)^2+\bar{b} \frac{2 r^{n+1}}{v^n(r-1)}(s-1).
\end{equation}
It can be checked that $Q(1)=4>0$ and $Q((r+1)/(r-1)) >0$ for $\bar{b}>0$ and $r>1$. Because $Q(s)$ is a quadratic function, $Q(s)$ always has $0$ or $2$  zeros between $1$ and $ (r+1)/(r-1)$. Based on this fact
and Lemma \ref{lem_iff_strong_stab}, we get that $ \hat{P}(s)$ is always strongly stabilizable for all $r>1$ and $\bar{b}>0$. Then by Lemma \ref{lem_iff_co_stabilizable}, we complete the proof.
\end{proof}

\subsection{Proof of Theorem \ref{coro_tasos}}
\label{proof_theorem_coro_tasos}

\begin{proof}[Proof of Theorem \ref{coro_tasos}]
   For the pair in \eqref{eq_zeng} with $r>1$, $\alpha_3=1$, $\alpha_4=-1$, $0<v<1$, and $b^{(3)}=b^{(4)}=0$, $S_3$ and $S_4$ can be reduced to the following two scalar time-delay systems: 
      \begin{equation}
\label{reduced_tasos}
\begin{aligned}
S_i &: \quad x_{t+1}^{(1)} = r x_{t}^{(1)} + \alpha_i v^{n} u_{t-n+1}, 
\end{aligned}
\end{equation}
where $i\in\{5,6\}$, $\alpha_5=1$ and  $\alpha_6=-1$. Therefore, the co-stabilization of $S_3$ and $S_4$ is equivalent to that of $S_5$ and $S_6$ in \eqref{reduced_tasos}. When taking $a=r$ and $b=v^n$ in Theorem \ref{thm_impossible_costabilization}, we know that there does not exist a co-stabilizing linear dynamic state feedback controller for:
     \begin{equation}
\label{reduced_scalar}
S_i : \quad x_{t+1}^{(1)} = r x_{t}^{(1)} + \alpha_iv^{n} u_{t}, 
\end{equation}
where $i\in\{7,8\}$, $\alpha_7=1$ and  $\alpha_{8}=-1$.
Next, we prove the conclusion by the contrapositive.
Assume that there exists a linear  dynamic state feedback controller 
\begin{equation}
\label{eq:controller_appendix}
\begin{aligned}
\z_{t+1} = \F \z_t + \G x_t,  
u_t = \Hb \z_t + \K x_t,
\end{aligned}
\end{equation}
that co-stabilizes the systems in \eqref{reduced_tasos}.

Introduce the delay-input-state
\[
\boldsymbol{\eta}_t :=
\begin{bmatrix}
u_{t-n+1} & u_{t-n+2} & \cdots & u_{t-1}
\end{bmatrix}^\top \in \mathbb{R}^{n-1}.
\]
Then the following  dynamics always satisfy
\begin{equation}
\label{delay_controller}
\boldsymbol{\eta}_{t+1} = \mathbf{S}\boldsymbol{\eta}_t + \mathbf{e} u_t,
\qquad
u_{t-n+1} = \mathbf{c}^\top \boldsymbol{\eta}_t,
\end{equation}
where
\[
\mathbf{S} :=
\begin{bmatrix}
0 & 1 & 0 & \cdots & 0\\
0 & 0 & 1 & \cdots & 0\\
\vdots & & & \ddots & \vdots\\
0 & 0 & 0 & \cdots & 1\\
0 & 0 & 0 & \cdots & 0
\end{bmatrix},
\quad
\mathbf{e} :=
\begin{bmatrix}
0\\ \vdots\\ 0\\ 1
\end{bmatrix},
\quad
\mathbf{c} :=
\begin{bmatrix}
1\\ 0\\ \vdots\\ 0
\end{bmatrix}.
\]
The delayed systems in \eqref{reduced_tasos} can therefore be written as
\begin{equation}
\label{eq:augmented_plant}
\begin{bmatrix}
x_{t+1}\\ \boldsymbol{\eta}_{t+1}
\end{bmatrix}
=
\begin{bmatrix}
r & \pm v^{n} \,\mathbf{c}^\top\\
\mathbf{0} & \mathbf{S}
\end{bmatrix}
\begin{bmatrix}
x_t\\ \boldsymbol{\eta}_t
\end{bmatrix}
+
\begin{bmatrix}
0\\ \mathbf{e}
\end{bmatrix}
u_t .
\end{equation}
Interconnecting \eqref{eq:augmented_plant} with the controller
\eqref{eq:controller_appendix} yields the following closed-loop autonomous LTI systems on the
state $\tilde{\z}_t := [x_t \;\boldsymbol{\eta}_t^\top \; \z_t^\top]^\top$  
\begin{equation}
\label{eq:cl_augmented}
\tilde{\z}_{t+1} 
=
\mathbf{A}_{\mathrm{cl},\pm}
 \tilde{\z}_t ,
\end{equation}
where the closed-loop matrix is
\begin{equation}
\label{eq:Acl_augmented}
\mathbf{A}_{\mathrm{cl},\pm}
=
\begin{bmatrix}
r & \pm v^{n}\,\mathbf{c}^\top & \mathbf{0}\\
\mathbf{e}\K & \mathbf{S} & \mathbf{e}\Hb\\
\G & \mathbf{0} & \F
\end{bmatrix}.
\end{equation}
By assumption, $\mathbf{A}_{\mathrm{cl},+}$ and $\mathbf{A}_{\mathrm{cl},-}$ are both stable.
 
Next, we consider the costabilization of the systems in \eqref{reduced_scalar}. Let $\boldsymbol{\delta}_t\in\mathbb{R}^{n-1}$ denote the controller-side delay state:
\[
\boldsymbol{\delta}_t :=
\begin{bmatrix}
u^{K}_{t-n+1} & u^{K}_{t-n+2} & \cdots & u^{K}_{t-1}
\end{bmatrix}^\top,
\]
where $u^{K}_t := \Hb \z_t + \K x_t$.
Then  \eqref{delay_controller} implies the realization
\begin{equation}
\label{eq:controller0}
\begin{aligned}
&\z_{t+1}  = \F \z_t + \G x_t, 
&u^K_t  = \Hb \z_t + \K x_t,\\
&\boldsymbol{\delta}_{t+1}  = \mathbf{S}\boldsymbol{\delta}_t + \mathbf{e} u^K_t, 
&u_t  = \mathbf{c}^\top \boldsymbol{\delta}_t .
\end{aligned}
\end{equation}
This is again a linear dynamic state feedback controller.

Combining the controller in \eqref{eq:controller0} and the non-delayed plant
\eqref{reduced_scalar}, the resulting closed-loop dynamics on $\hat{\z}_t := [x_t \;\boldsymbol{\delta}_t^\top \; \z_t^\top]^\top$
 are
\begin{equation}
\label{eq:cl_nodelay}
 \hat{\z}_{t+1} 
=
\mathbf{A}_{\mathrm{cl},\pm}
 \hat{\z}_t .
\end{equation}
Comparing \eqref{eq:cl_nodelay} with the closed-loop dynamics in \eqref{eq:cl_augmented}, we see that the closed-loop matrices are identical up to relabeling
$\boldsymbol{\eta}_t \leftrightarrow \boldsymbol{\delta}_t$.
Hence, for each choice of sign $\pm$, the closed-loop matrix in
\eqref{eq:cl_nodelay} is stable.

Therefore, the co-stabilization of the systems in \eqref{reduced_tasos} implies the co-stabilization of the systems in 
\eqref{reduced_scalar}. In other words, if there does not exist a co-stabilizing linear dynamic state feedback controller for \eqref{reduced_scalar}, there also does not exist a co-stabilizing linear dynamic state feedback controller for \eqref{reduced_tasos}.
 
 In conclusion, there does not exist a co-stabilizing linear dynamic state feedback controller for $S_3$ and $S_4$ when the conditions in Theorem \ref{coro_tasos} hold.
\end{proof}

\balance
\bibliographystyle{IEEEtran}
\bibliography{references}

\end{document}